\documentclass[11pt]{article}

\usepackage{fullpage}
\usepackage[procnumbered,ruled,vlined,linesnumbered,algo2e]{algorithm2e}

\usepackage{xcolor}
\usepackage{pdfsync}
\usepackage{commath}

\usepackage{comment}

\def\showauthornotes{1}

\def\showdraftbox{0}

\usepackage{fancybox}

\usepackage{amsmath,amssymb,amsthm,amstext,amsfonts,bbm,algorithm,algorithmicx,graphicx,xspace,nicefrac}
\usepackage{color,stmaryrd,enumerate,latexsym,bm,amsfonts,wrapfig,verbatim,tabularx,textcomp,
subfig}

\usepackage{amsfonts}
\usepackage{comment} 
\usepackage{epsfig} 
\usepackage{latexsym,nicefrac,bbm}
\usepackage{xspace}
\usepackage{color,fancybox,graphicx,url}
\usepackage{enumitem}
\usepackage{booktabs}
\usepackage{commath}
\usepackage{mdframed}
\usepackage{pdfsync}

\usepackage{thm-restate}

\usepackage{fancybox}
\newenvironment{fminipage}%
  {\begin{Sbox}\begin{minipage}}%
  {\end{minipage}\end{Sbox}\fbox{\TheSbox}}

\newcommand{\defeq}{\stackrel{\textup{def}}{=}}

\definecolor{ForestGreen}{rgb}{0.1333,0.5451,0.1333}
\definecolor{DarkRed}{rgb}{0.8,0,0}
\definecolor{Red}{rgb}{1,0,0}
\usepackage[linktocpage=true,
pagebackref=true,colorlinks,
linkcolor=DarkRed,citecolor=ForestGreen,
bookmarks,bookmarksopen,bookmarksnumbered]
{hyperref}

\usepackage{cleveref}
\usepackage{mathtools}
\usepackage{thmtools}

\declaretheorem[numberwithin=section]{theorem}
\declaretheorem[numberlike=theorem]{lemma}

\declaretheorem[numberlike=theorem]{proposition}

\declaretheorem[numberlike=theorem]{corollary}

\declaretheorem[numberlike=theorem,style=definition]{definition}

\theoremstyle{remark}
\newtheorem*{rem*}{Remark}

\newcommand{\nfrac}[2]{\nicefrac{#1}{#2}} \def\abs#1{\left| #1
  \right|} 

\newcommand{\marginlabel}[1]%
{\mbox{}\marginpar{\it{\raggedleft\hspace{0pt}#1}}}

\newcommand\poly{{\textrm{poly}}}  

\newcommand{\support}{\textrm{support}}

\definecolor{Mygray}{gray}{0.8}

\ifcsname ifcommentflag\endcsname\else
\expandafter\let\csname ifcommentflag\expandafter\endcsname
\csname iffalse\endcsname
\fi

\ifnum\showauthornotes=1
\newcommand{\todo}[1]{\colorbox{Mygray}{\color{red}\parbox{\textwidth}{#1}}}
\else
\newcommand{\todo}[1]{}
\fi

\ifnum\showauthornotes=1
\newcommand{\Authornote}[2]{{\sf\small\color{red}{[#1: #2]}}}
\newcommand{\Authoredit}[2]{{\sf\small\color{red}{[#1]}\color{blue}{#2}}}
\newcommand{\Authorcomment}[2]{{\sf \small\color{gray}{[#1: #2]}}}
\newcommand{\Authorfnote}[2]{\footnote{\color{red}{#1: #2}}}
\newcommand{\Authorfixme}[1]{\Authornote{#1}{\textbf{??}}}
\newcommand{\Authormarginmark}[1]{\marginpar{\textcolor{red}{\fbox{
#1:!}}}}
\else
\newcommand{\Authornote}[2]{}
\newcommand{\Authoredit}[2]{}
\newcommand{\Authorcomment}[2]{}
\newcommand{\Authorfnote}[2]{}
\newcommand{\Authorfixme}[1]{}
\newcommand{\Authormarginmark}[1]{}
\fi

\newlength{\pgmtab}  
 {
	\begin{enumerate}}{\end{enumerate}}

\def\qedsketch{\ifmmode\Box\else{\unskip\nobreak\hfil
\penalty50\hskip1em\null\nobreak\hfil$\Box$
\parfillskip=0pt\finalhyphendemerits=0\endgraf}\fi}

\newlength{\tpush}
\newcommand{\handout}[5]{
   \noindent
   \begin{center}
   \framebox{ \vbox{ \hbox to \textwidth { {\bf \coursenum\ :\  \coursename} \hfill #5 }
       \vspace{3mm}
       \hbox to \textwidth { {\Large \hfill #2  \hfill} }
       \vspace{1mm}
       \hbox to \textwidth { {\it #3 \hfill #4} }
     }
   }
   \end{center}
   \vspace*{4mm}
   \newcommand{\lecturenum}{#1}
   \addcontentsline{toc}{chapter}{Lecture #1 -- #2}
}

\ifnum\showdraftbox=1

\else

\fi

\allowdisplaybreaks

\def\sparsifier{NMC sparsifier}
\def\trim{\mathrm{trim}}
\def\shave{\mathrm{shave}}

\def\defeq{:=}

\def\vol{\mathrm{vol}}

\def\abs#1{\left|#1  \right|}

\newcommand{\DSprune}{\mathcal{D}}
\newcommand{\counterDeletion}{\#\mathrm{del}}

\renewcommand\leq{\leqslant}
\renewcommand\geq{\geqslant}

\newcommand\Otil{\tilde{O}}

\newcommand{\eps}{\varepsilon}

\renewcommand{\setminus}{\backslash}

\DontPrintSemicolon
\SetKwProg{Procedure}{Procedure}{}{}
\SetProcNameSty{textsc}
\SetFuncSty{textsc}
\SetKw{KwAnd}{and}
\SetKw{KwDo}{do}

\SetCommentSty{mycommfont}

\SetKwFunction{Query}{Query}
\SetKwFunction{AddTerminal}{AddTerminal}
\SetKwFunction{Initialize}{Initialize}
\SetKwFunction{Insert}{Insert}

\global\long\def\Otil{\tilde{O}}
\global\long\def\poly{\mathrm{poly}}

\usepackage{wrapfig}
\newcommand{\erclogowrapped}[1]{%
\setlength\intextsep{0pt}%
\begin{wrapfigure}[3]{r}{#1*\real{1.1}}%
\includegraphics[width=#1]{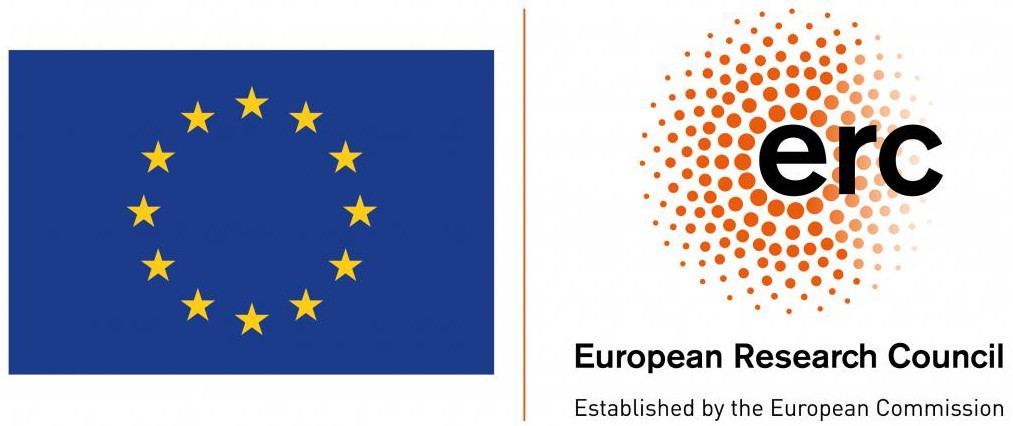}%
\end{wrapfigure}%
}

\makeatletter
\def\blfootnote{\xdef\@thefnmark{}\@footnotetext}
\makeatother

\begin{document}

\title{Fully Dynamic Exact Edge Connectivity in Sublinear Time}
\author{Gramoz Goranci\thanks{Institute for Theoretical Studies, ETH Zurich. This project is supported by Dr. Max R\"ossler, the Walter Haefner Foundation and the ETH Z\"urich Foundation. Part of this work was done while the author was at the University of Glasgow. email: \texttt{gramoz.goranci@eth-its.ethz.ch 
}} \and Monika Henzinger\thanks{Faculty of Computer Science, University of Vienna. This project has received funding from the European Research Council (ERC) under the European Union's Horizon 2020 research and innovation programme (Grant agreement No. 101019564 ``The Design of Modern Fully Dynamic Data Structures (MoDynStruct)'' and from the Austrian Science Fund (FWF) project ``Fast Algorithms for a Reactive Network Layer (ReactNet)'', P~33775-N, with additional funding from the \textit{netidee SCIENCE Stiftung}, 2020--2024.
email: \texttt{monika.henzinger@univie.ac.at}} \and Danupon Nanongkai\thanks{Max Planck Institute for Informatics \& University of Copenhagen \& KTH. This project has received funding from the European Research Council (ERC) under the European Union's Horizon 2020 research and innovation programme under grant agreement No 715672. Nanongkai was also partially supported by the Swedish Research Council (Reg. No. 2019-05622). email: \texttt{danupon@mpi-inf.mpg.de}} \and Thatchaphol Saranurak\thanks{University of Michigan. email: \texttt{thsa@umich.edu} } \and Mikkel Thorup\thanks{University of Copenhagen. This project is supported by Investigator Grant 16582, Basic Algorithms Research Copenhagen (BARC), from the VILLUM Foundation. email: \texttt{mikkel2thorup@gmail.com}
} \and Christian Wulff-Nilsen\thanks{University of Copenhagen. This project is supported by the Starting Grant 7027-00050B from the Independent Research Fund Denmark under the Sapere Aude research career programme. email: \texttt{koolooz@di.ku.dk}} }

\maketitle
\blfootnote{\erclogowrapped{5\baselineskip}\\}

\pagenumbering{gobble}

\begin{abstract}
Given a simple $n$-vertex, $m$-edge graph $G$ undergoing
edge insertions and deletions, we give two new fully dynamic algorithms for exactly
maintaining the edge connectivity of $G$ in $\Otil(n)$ worst-case update time and
$\Otil(m^{1-1/31})$ amortized update time, respectively. Prior to our work,
all dynamic edge connectivity algorithms assumed bounded edge connectivity, guaranteed approximate solutions, or were restricted to edge insertions only. Our results answer
in the affirmative an open question posed by Thorup~{[}Combinatorica'07{]}.
\end{abstract}


\pagenumbering{arabic}
\newpage

\section{Introduction}

The edge connectivity of an undirected, unweighted graph $G=(V,E)$ is the minimum number of edges whose removal disconnects the graph $G$. Finding the edge connectivity of a graph is one of the cornerstone problems in combinatorial optimization and dates back to the work of Gomory and Hu~\cite{GomoryH61} in 1961. Since then, a large body of research work has dealt with the question of obtaining faster algorithms for this problem in the classical sequential setting~\cite{ford1962flows,NagamochiI92_mincut,HaoO94,Gabow95,StoerW97,KargerS96,Karger00,KawarabayashiT19,HenzingerRW20,LiP20,BhardwajLS20,GawrychowskiMW20,MukhopadhyayN20,gawrychowski2021note,li2021deterministic}. This line of work culminated in a breakthrough result by Kawarabayashi and Thorup~\cite{KawarabayashiT19} in 2015 who obtained a \emph{deterministic} algorithm that runs in $\Otil(m)$\footnote{We use $\Otil(\cdot)$ to hide poly-logarithmic factors.} time on a $n$-vertex, $m$-edge graph, which was later improved by Henzinger, Rao, and Wang~\cite{HenzingerRW20} to $O(m \log^2 n \log \log^2 n)$. Edge connectivity has also been extensively studied in various models of computation including the parallel model  \cite{Karger93,Geissmann2018parallel,lopez2021work},
the distributed models \cite{pritchard2011fast,ghaffari2013distributed,nanongkai2014almost,ghaffari2018congested,daga2019distributed,parter2019small,ghaffari2020massively,ghaffari2020faster,Dory2021distributed}, the semi-streaming model \cite{ahn2009graph,MukhopadhyayN20,assadi2021simple},
and several query models \cite{rubinstein2017computing,MukhopadhyayN20,lee2020cut}.
All these models admit non-trivial, if not near-optimal, algorithms
for \emph{exactly} computing edge connectivity.

We study edge connectivity in the \emph{fully dynamic} setting, where the underlying graph $G$ undergoes edge insertions and deletions, known as edge \emph{updates}, and the goal is to maintain the edge connectivity of $G$ after each update with as small update time as possible. In contrast to the long line of research work in other computational models, the only known algorithm for the fully dynamic edge connectivity problem is the trivial solution of recomputing the edge connectivity from scratch after each update, which costs $\tilde \Theta(m)$ time per update. Thorup~\cite{thorup2007fully} introduced this problem and gave a fully dynamic edge connectivity algorithm that supports fast updates as long as the edge connecitvity is upper bounded by some parameter $\eta$, where $\eta$ is a small polynomial in $n$. Concretely, his algorithm achieves $\Otil(\eta^{29/2} \sqrt{n})$ worst-case time per edge update, and thus is slower than the trivial algorithm whenever $\eta =\tilde \Omega( m^{2/29} n^{-1/29})$. In spite of dynamic graph algorithms being a flourishing research field, prior to our work, there has been no progress on the fully dynamic edge connectivity problem in the last 15 years. 

In this paper we give the first solutions with $o(m)$ update time, answering in the affirmative an open question posed by Thorup~\cite{thorup2007fully} of whether this is possible. More concretely, we show the following two results.

\begin{theorem}
\label{thm:main n} Given an undirected, unweighted $n$-vertex, $m$-edge graph $G=(V,E)$, there is a fully dynamic randomized algorithm that processes an online sequence of edge insertions or deletions and maintains the edge connectivity of $G$ in 
$\Otil(n)$ worst-case update time with high probability.
\end{theorem}

The above randomized algorithm works against an \emph{adaptive} adversary and achieves \emph{sub-linear} update time as long as $m=\Omega(n^{1+\epsilon})$ where $\epsilon$ where is some positive constant. We complement both points of this result by designing a second algorithm that is (i) deterministic and (ii) achieves sub-linear update times regardless of graph density.


\begin{theorem}
\label{thm:main det} Given an undirected, unweighted $n$-vertex, $m$-edge graph $G=(V,E)$, there is a fully dynamic \emph{deterministic} algorithm that processes an online sequence of edge insertions or deletions and maintains the edge connectivity of $G$ in $\Otil(m^{29/31}n^{1/31}) = \Otil(m^{1-1/31})$ amortized update time.
\end{theorem}

Both algorithms can also report the edges on a cut 
that attains the edge connectivity of $G$ in time nearly proportional to the edge connectivity, with the caveat that the algorithm
from Theorem \ref{thm:main n} then only works against an \emph{oblivious} adversary.

\subsection{Our Techniques}

In this section, we discuss the main obstacles that explain the lack of progress on fully dynamic exact edge connectivity algorithms
and point out the key ideas that enable us to achieve our results. 

Before 2015, all near-linear time algorithms for exact
edge connectivity reduced to computing a minimum 2-respecting cut of
a spanning tree~\cite{Karger00}. This problem involves setting up a sophisticated dynamic programming solution, and the question of whether this solution admits fast dynamic algorithms remains a notoriously
hard open problem. For the somewhat easier problem of maintaining a minimum 1-respecting cut of a spanning tree~\cite{frederickson1985data,alstrup2005maintaining,thorup2007fully}, state-of-the-art dynamic algorithms allow us to solve the problem efficiently.
In fact, this is used as a key subroutine in Thorup's algorithm \cite{thorup2007fully}
that $(1+\epsilon)$-approximates edge connectivity in $\Otil(\sqrt{n})$ worst-case update time.

In a breakthrough work, Kawarabayashi and Thorup \cite{KawarabayashiT19} in 2015
showed a largely different approach to tackling the edge connectivity problem. Their key insight
was to introduce a notion of sparsification for edge connectivity in a subtle way: given an undirected, unweighted $n$-vertex $G$, one can contract
edges of $G$ and obtain a graph $G'$ such that (i) $G'$ has only $O(n)$ edges and (ii) $G'$
preserves all \emph{non-singleton} minimum cuts of $G$.\footnote{A non-singleton minimum cut is a minimum cut where both sides of the cut contain at least $2$ vertices.} Throughout, we call the graph $G'$ a \emph{non-singleton minimum cut sparsifier} (abbrv. \emph{NMC sparsifier}). Since maintaining singleton
minimum cuts boils down to maintaining the minimum degree of $G$, and the latter can be easily achieved as $G$ undergoes edge updates, we can focus our attention to designing fully dynamic algorithms for maintaing the NMC sparsifier $G'$.


In the insertions-only setting, Goranci, Henzinger, and Thorup~\cite{goranci2018incremental} observed that
an NMC sparsifier interacts well with edge insertions, as it satisfies a certain composability property. Specifically, they showed that given an NMC sparsifier $G'$ of a graph $G$ and an edge insertion $e$ to $G$, the graph $G' \cup \{e\}$ remains an NMC sparsifier of $G \cup \{e\}$. This was in turn combined with Henzinger's insertions-only algorithm~\cite{Henzinger97} for maintaining small edge connectivity and $k$-connectivity certificates. Periodically invoking a static algorithm for computing NMC sparsifier in a black-box manner then led to a dynamic algorithm with poly-logarithmic amortized update time per edge insertion. 

We may be tempted to employ the same approach for handling edge deletions. However, a short afterthought reveals that the crucial composability property we used for edge insertions completely fails for edge deletions. This suggests that restricting to edge deletions does not seem to help in the dynamic NMC sparsifier problem, so in this work we refocus our attention to the fully dynamic setting.

We devise two new fully dynamic NMC sparsifier algorithms which lead us to Theorems \ref{thm:main n} and \ref{thm:main det},
respectively. The first one is randomized and is based on a dynamic variant of the random $2$-out contraction technique leveraged by Ghaffari, Nowicki, and Thorup~\cite{ghaffari2020faster} (\Cref{sec:randomized}), whereas the second one is deterministic and builds upon the expander decomposition-based approach for computing edge connectivity by Saranurak~\cite{Saranurak21} (\Cref{sec:det}). We note that the original construction of NMC sparsifiers~\cite{KawarabayashiT19} is already quite involved in the static setting and seems difficult to adapt in the dynamic setting. In the paragraphs below, we give a succinct summary of the technical ideas behind both of our algorithms.

Key to our randomized algorithm for dynamically maintaining an NMC sparsifier is the following construction: given a graph $G$, for each vertex $u \in V$, sample two incident edges to $u$ independently, with replacement, and contract them to obtain the graph $G'$, which we call a \emph{random $2$-out contraction} of $G$. Despite the fact that $G'$ is known to have only $O(n/\delta)$ vertices~\cite{ghaffari2020faster}, where $\delta$ is the minimum degree of $G$, the number of edges in $G'$ could  potentially still be large, say $\omega(n)$, and thus inefficient for our purposes. The main technical component of our dynamic algorithm is to efficiently maintain a \emph{sparse} $\delta$-connectivity certificate $H'$ of $G'$. Towards achieving this goal, we have to deploy a variety of algorithmic tools from sequential, parallel, and streaming algorithms, namely (i) sequential and parallel constructions of $k$-connectivity certificates~\cite{nagamochi1992linear, daga2019distributed}, and (ii) constructing spanning forests in sub-linear time using linear $\ell_0$-sampling sketches~\cite{cormode2014unifying, kapron2013dynamic}. A more detailed description of the algorithm can be found in~\Cref{sec:randomized}.

Our deterministic algorithm follows the now-widespread and powerful algorithmic approach of employing \emph{expander decompositions} for solving graph-based optimization problems. At a high level, an expander decomposition is a partitioning of a graph into \emph{well-connected} clusters, whose expansion is controlled by a parameter $\phi \in (0,1)$, such that there are few inter-cluster edges left, say roughly $\phi m$. If $\phi \approx \delta^{-1}$, then Saranurak~\cite{Saranurak21} recently showed that contracting a carefully chosen \emph{vertex subset} of each expander in the decomposition leads to a NMC sparsifier $G'$. Our main technical contribution is a simple, deletions-only algorithm for maintaining an expander decomposition (based on expander prunning~\cite{SaranurakW19}), which in turn leads to a deletions-only algorithm for maintaining the NMC sparsifier $G'$. While expander pruning has been already used for dynamically maintaining other graph-based properties~\cite{GoranciRST21,BernsteinBGNSS022}, we believe that our construction is one of the simplest and may prove useful in future applications. We extend our deletions-only NMC algorithm to a fully dynamic one by keeping edge insertions ``on the side'' and rebuilding periodically. Finally, for achieving our claimed sub-linear update time, our NMC sparsifier algorithm is run in ``parallel'' with the exact fully dynamic edge connectivity algorithm of~\cite{thorup2007fully} which returns correct answers only for small edge connectivity. For further details we point the reader to~\Cref{sec:det}.

\subsection{Related Work}

The study of \emph{approximation} algorithms for the fully dynamic edge connectivity problem was initiated by Karger~\cite{Karger94} who gave a randomized algorithm that maintains a $\sqrt{1+2/\epsilon}$ to edge connectivity in $\Otil(n^{1/2 +\eps})$ expected amortized time per edge operation. Karger and Thorup \cite{thorup2000dynamic}
showed a fully dynamic algorithm that $(2+\epsilon)$-approximates
edge connectivity in $\Otil(1)$ amortized updated time. Thorup~\cite{thorup2007fully} improved the approximation factor to $(1+\epsilon)$ at the cost of increasing the update time to $\Otil(\sqrt{n})$. However, his running time guarantees are worst-case instead of amortized.

{} 

Prior to our work, all known fully dynamic algorithms for \emph{exactly} computing the edge connectivity $\lambda(G)$ of $G$ take sub-linear time only when $\lambda(G)$ is small. In the same work~\cite{thorup2007fully}, Thorup also showed an exact fully dynamic algorithm with $\Otil(\lambda(G)^{29/2}\sqrt{n})$
worst-case update time, which is sub-linear whenever $\lambda(G)=o(m^{2/29}n^{-1/29})$.\footnote{Nevertheless, Thorup's result does not assume that $G$
is an undirected, unweighted graph.} For $\lambda(G)$ being a small constant, edge connectivity can be maintained in  $\Otil(1)$ amortized update time. Specifically, there were a series of refinements in the literature for designing fully dynamic algorithms for graph connectivity (i.e., checking whether $\lambda(G) \geq 1$)~\cite{frederickson1985data,henzinger1999randomized,holm2001poly,kapron2013dynamic,nanongkai2017dynamic,chuzhoy2020deterministic} and $2$-edge connectivity (i.e., checking whether $\lambda(G)\geq 2$)~\cite{frederickson1997ambivalent,henzinger1997fully,holm2001poly,holm2018dynamic}. When the underlying dynamic graph is guaranteed to remain \emph{planar} throughout the whole sequence of online updates, Lacki and Sankowski~\cite{lkacki2011min} gave an algorithm with $\Otil(n^{5/6})$ worst-case update time per operation.

Partially dynamic algorithms, i.e., algorithms that are restricted to either edge insertions or deletions only, have also been studied in the context of exact maintenance of edge connectivity. Henzinger~\cite{Henzinger97} designed an insertions-only algorithm with $O(\lambda(G) \log n)$ amortized update time. Recently, Goranci, Henzinger, and Thorup
\cite{goranci2018incremental} showed how to improve the update time to $O(\log^3 n \log \log^ 2 n)$, thus removing the dependency on edge connectivity from the running time.


To summarize, all previous dynamic edge connectivity algorithms either maintain an approximation to $\lambda(G)$, require that $\lambda(G)$ is small, handle edges insertions only, or are restricted to special family of graphs such as planar graphs. Hence, our results are the
first fully dynamic exact edge connectivity algorithms that achieve sub-linear update times on general graphs. 

\section{Preliminaries}

Let $G=(V,E)$ be an $n$-vertex, $m$-edge graph. For a set $U \subseteq V,$ the \emph{volume} of $U$ in $G$ is defined as $\vol_G(U) := \sum_{u \in U} d(u)$, where $d(u)$ denotes the degree of $u$ in $G$. Let $\delta = \min_{u \in U} \{d(u)\}$ denote the \emph{minimum degree} in $G$. A \emph{cut} is a subset of vertices $C \subseteq V$ where  $\min\{\abs{C}, \abs{V \setminus C}\} \geq 1$. A cut $C$ is \emph{non-singleton} iff $\min\{\abs{C}, \abs{V \setminus C}\} \geq 2$. For two disjoint sets $S,T \subseteq V$, let $E(S,T) \subseteq E$ be the set of edges with one endpoint in $S$ and the other in $T$. Let $\partial(S) \defeq E(S, V \setminus S)$. The \emph{edge connectivity} in $G$, denoted by $\lambda(G)$, is the cut $C$ that minimizes $\abs{\partial(C)}$.

It is a well-known fact that edge connectivity can be computed in near-linear time on the number of edges.
\begin{theorem}[\cite{Karger00}] \label{thm: staticEdgeConnectivity}
	Let $G=(V,E)$ be a weighted, undirected graph with $m$ edges. There is an algorithm that computes the edge connectivity $\lambda(G)$ of $G$ in $\Otil(m)$ time.
\end{theorem}

\section{Randomized Algorithm with $\Otil(n)$ Update Time}
\label{sec:randomized}

In this section we prove \Cref{thm:main n}. Our algorithm requires several tools from different works and we review them below.

\subsection{Algorithmic Tools}

\paragraph{Random $2$-Out Subgraph and Contraction.}

Let $G=(V,E)$ be a undirected simple graph. Let $R=(V,E')$ be a
\emph{random $2$-out subgraph} of $G$ which is obtained from $G$ using the following procedure
\begin{itemize}
\itemsep0em
\item Set $R \gets (V,E')$, where $E' \gets \emptyset$.
\item For each $u \in V$: 
\begin{itemize}
\itemsep0em
\item[]$\triangleright~$ Sample from $E$ two incident
edges $(u,v_{1})$, $(u,v_{2})$ independently, with replacement.
\item[]$\triangleright~$ Add $(u,v_{1})$ and $(u,v_{2})$ to $E'$.
\end{itemize}
\end{itemize}
The graph $G'=G/R$ obtained
by contracting all edges of $R$ is called a \emph{random $2$-out
contraction}. Ghaffari, Nowicki and Thorup~\cite{ghaffari2020faster} showed that a random 2-out contraction reduces the number of nodes to $O(n/\delta)$ whp, while preserving any fixed non-singleton nearly minimum cut with constant probability, as in the theorem below.
\begin{theorem}
[Theorem 2.4 of \cite{ghaffari2020faster}]\label{thm:random-2-out}A
random $2$-out contraction of a graph with $n$ vertices and minimum
degree $\delta$ has $O(n/\delta)$ vertices, with high probability,
and preserves any fixed non singleton $(2-\epsilon)$ minimum cut,
for any constant $\epsilon\in(0,1]$, with some constant probability
$p_{\epsilon}>0$.
\end{theorem}

\paragraph{Sequential and Parallel Constructions of $k$-Connectivity Certificates.}

Given a graph $G=(V,E)$, a \emph{$k$-connectivity certificate} $H$ of $G$ is a subgraph of $G$ that preserves all cuts of size at most $k$. Concretely, for any vertex set $S \subseteq V$, 
$\min\{k, |E_H(S,V\setminus S)|\} = \min\{k, |E_G(S,V\setminus S)|\}$. Nagamochi and Ibaraki~\cite{nagamochi1992linear} designed a sequential algorithm for computing a \emph{sparse} $k$-connectivity certificate in linear time. Below we shall review an algorithm that does not run in linear time but it's simpler and suffices for our purposes.


\begin{algorithm}
\textbf{Input: }A graph \textbf{$G=(V,E)$ }with $n$ vertices, and
a parameter $k$. \\
\textbf{Output: } A $k$-connectivity certificate $H$ of $G$.
\begin{enumerate}
\itemsep0em
\item Set $G_{1}\gets G$. 
\item For $i \gets 1,\ldots,k$:
\begin{enumerate}
\itemsep0em
\item Find a spanning forest $F_{i}$ of $G_{i}$.
\item Set $G_{i+1}\gets G_{i}\setminus F_{i}$.
\end{enumerate}
\item Return $H=\cup_{i=1}^{k}F_{i}$.
\end{enumerate}
\caption{\label{alg:seq certificate}A sequential algorithm for computing a $k$-edge connectivity certificate}
\end{algorithm}
\begin{theorem}[\cite{nagamochi1992linear}]
\label{thm:seq certificate}Given a graph $G=(V,E)$ with $n$ vertices
and an integer parameter $k \geq 1$, \Cref{alg:seq certificate} returns
a $k$-connectivity certificate $H$ of $G$ containing $O(nk)$ edges.
\end{theorem}
Observe that \Cref{alg:seq certificate} for constructing a $k$-connectivity certificate computes $k$ \emph{nested}
spanning forests. When $k$ is large, this long chain of dependency
is too inefficient in the dynamic setting. Although \Cref{alg:seq certificate} will prove useful at some point in our final construction, we need to bypass this dependency issue. To this end, we will exploit an alternative
$k$-connectivity certificate construction by Daga et al.~\cite{daga2019distributed} which was
developed in the context of distributed and parallel algorithms. We describe this construction in \Cref{alg:para certificate}. The key advantage of this algorithm
is that it reduces the $k$-connectivity certificate problem to $O(k/\log n)$
instances of $k'$-connectivity certificate where $k'=O(\log n)$. This suggests that we can afford using algorithms even with polynomial dependency on $k'$ since $k'$ is logarithmic on the number of nodes.

\begin{algorithm}
\textbf{Input: }A graph \textbf{$G=(V,E)$ }with $n$ vertices, and
a parameter $k$. \\
\textbf{Output: }A $k$-connectivity certificate $H$ of $G$.
\begin{enumerate}
\itemsep0em
\item Choose $c \gets k/(4\tau\ln n)$ where $\tau$ is a big enough constant and $c$ is an integer. Let $k' \gets \left\lceil 6\tau\log n\right\rceil$.
\item Randomly color each edge of $G$ using colors from $\{1,\dots,c\}$.
Let $E_{i}$ be a set of edges with color $i$. Let $G_{i}=(V,E_{i})$.
\item For $i \gets 1, \ldots, c$:
\begin{enumerate}
\itemsep0em
    \item Apply \Cref{alg:seq certificate} to compute a $k'$-connectivity certificate $H_i$ of $G_{i}$.
\end{enumerate} 
\item Return $H=\cup_{i=1}^{c}H_{i}$.
\end{enumerate}
\caption{\label{alg:para certificate}A parallel algorithm for computing a $k$-edge
connectivity certificate}
\end{algorithm}
\begin{theorem}
[Section 3 of \cite{daga2019distributed}]\label{thm:para certificate}Given
a graph $G=(V,E)$ with $n$ vertices and an integer parameter $k$, 
\Cref{alg:para certificate} returns a subgraph $H$ of $G$ such that $\textup{(a)}$ $H$ contains $O(nk)$ edges, and $\textup{(b)}$ $H$ is a $k$-connectivity
certificate of $G$ with high probability. 
\end{theorem}

\paragraph{Spanning Forest From $\ell_{0}$-Sampling Sketches.}
A well-known tool in the streaming algorithms literature is the $\ell_{0}$-sampling technique. This tool is particularly useful in the context of the following natural problem: given a $N$-dimensional vector $x$, we would like to construct a data structure supporting the following:
\begin{itemize}
\itemsep0em
    \item if $x = 0$, it reports \texttt{null},
    \item if $x \neq 0$, then it returns some $e \in \support(x)$ with high probability,
\end{itemize}
where $\support(x)$ consists of all \emph{non-zero} entries of $x$. In our concrete application, $x$ will correspond to a vertex of the graph and we will store the edges incident to $x$ in such a data structure, i.e.,~$\support(x)$ will be the edges incident to $x$. By itself such a data structure is trivial: just keep an adjacency list for every vertex $x$. However, 
for the concrete implementation of the above data structure, we use a linear function, widely reffered to as a \emph{linear sketching transformation} (aka linear sketch) $\Pi$, such that given $\Pi(x)$ for $x$ and $\Pi(y)$ for $y$, we get that $\Pi(x) + \Pi(y)$ gives a data structure that (1) returns an edge incident to $x$ or $y$ (but not to both) and (2) $\Pi(x) + \Pi(y)$ can be computed in $\tilde O(1)$ time.

More formally, in the theorem below we present the main result relating the $\ell_0$-sampling technique to the above data structure problem and focusing on running time instead of space guarantees.
\begin{theorem}[\cite{cormode2014unifying}]
\label{thm:sketch}For any dimension parameter $N$, there is a randomized algorithm for constructing a representation of a linear sketching transformation $\Pi:\mathbb{R}^{N}\rightarrow\mathbb{R}^{O(\log^{3}N)}$
in $\Otil(1)$ time, such that for any vector $x \in \mathbb{R}^{N}$,
\begin{enumerate}
\itemsep0em
    \item we can compute $\Pi (x)$ in $\Otil(|\support(x)|)$ time, and
    \item given $\Pi (x)$, we can obtain a \emph{non-zero} entry of $x$ in $\Otil(1)$ time with high probability.
\end{enumerate}
\end{theorem}
The sketch $\Pi$ can be represented using simple hash functions (see \cite{cormode2014unifying}) and avoids explicitly representing $\Pi$ as a matrix of size $O(\log^3N)\times N$.
This is why it only takes $\Otil(1)$ to initialize $\Pi$.

We call $\Pi (x)$ an \emph{$\ell_{0}$-sampling sketch}
of $x$. Our algorithm will maintain the $\ell_{0}$-sampling sketch
of the row-vectors of the signed vertex-edge incidence matrix of a graph $G$, which is defined as follows. Given a graph $G=(V,E)$ with $n$ vertices, the \emph{signed vertex-edge incidence
matrix} $B\in\{-1,0,1\}^{n\times\binom{n}{2}}$ of $G$ is
\[
B_{u,(v,w)}:=\begin{cases}
1 & (v,w)\in E\text{ and }u=v\\
-1 & (v,w)\in E\text{ and }u=w\\
0 & \text{otherwise.}
\end{cases}
\]
Let $b_{u}\in\{-1,0,1\}^{\binom{n}{2}}$ denote the $u$-th row of
$B$. We observe that one can efficiently compute and update the sketches $\Pi(b_u)$ for all $u \in V$.

\begin{proposition}
\label{prop:update sketch}Given a $n$-vertex graph $G$ with $m$ edges, there is an algorithm to compute a linear transformation $\Pi(b_u)$ for all $u \in V$ in
$\Otil(m)$ time. Upon an edge insertion or deletion in $G$, one can update the sketches $\Pi(b_u)$ for all $u \in V$ in $\Otil(1)$ time.
\begin{proof}
When computing $\Pi(b_u)$,
we only spend $\Otil(|\support(b_u)|)$ time by \Cref{thm:sketch} (1).
Since the incidence matrix $B$ contains only $\sum_u|\support(b_u)| = O(m)$ non-zeros, the first claim of the proposition follows. The second claim holds since (i) each edge update affects only two entries of $B$ and (ii) $\Pi$ is a linear transformation. 
Concretely, let $(u,v)$ be the updated edge and let $e_u$ and $e_v$ be the elementary unit vectors with the non-zero entry only at $u$ and $v$, respectively. 
We start by computing $\Pi(e_u)$ and $\Pi(e_v)$ in $\Otil(1)$ time, and then proceed to evaluating $\Pi(b_u) \pm \Pi(e_u)$ and $\Pi(b_v) \mp \Pi(e_v)$ in $\Otil(1)$ time, where the sign depends on whether $(u,v)$ is inserted or deleted.
\end{proof}
\end{proposition}

The sketches $\Pi(b_u)$ for all $u \in V$ are particularly useful since for any given any set $S\subseteq V$, they allow us to obtain an edge crossing the cut $(S,V\setminus S)$ in $\Otil(|S|)$ time. This is better than the trivial approach which scans through all edges incident to $S$ and takes $O(|S|^{2} + |E(S, V \setminus S)|)$ time in the worst case. 
More formally, let $b_S := \sum_{u \in S} b_u$ for any vertex set $S \subseteq V$. Now, by \Cref{thm:sketch} each $\Pi(b_u)$ can be queried in $\Otil(1)$ time. Using the linearity of $\Pi$, we compute $\Pi(b_S) = \sum_{u \in S} \Pi(b_u)$ in $\Otil(|S|)$ time. Observing that non-entries of $b_S$ correspond exactly to the edges crossing the cut $(S, V \setminus S)$, we can obtain one of these edges from $\Pi(b_S)$ in $\Otil(1)$ time. This observation has proven useful for building a spanning forest of a graph $G$ in $\Otil(n)$ time\footnote{Note that this sub-linear in the size of the graph since $G$ has $m$ edges} when give access to the sketches $\Pi(b_u)$, $u \in V$. More precisely, it is implicit in previous works~\cite{ahn2012analyzing,kapron2013dynamic} that if $O(\log n)$ \emph{independent} copies of the sketches are maintained, then a spanning forest can be constructed using Boruvka's algorithm, as summarized in the theorem below.

\begin{theorem}[\cite{ahn2012analyzing,kapron2013dynamic}]
\label{thm:forest from sketch}Let $G=(V,E)$ be an $n$-vertex graph. Let $\Pi_1,\ldots,\Pi_{O(\log n)}$
be linear transformations for the $\ell_{0}$-sampling sketch from
\Cref{thm:sketch} generated independently. Let $\Pi_i(b_u)$ be the sketches for all vertices $u\in S$ and all $i=1,\ldots, O(\log n)$. Then there is an algorithm to construct a spanning forest of $G$
in $\Otil(n)$ time with high probability.
\end{theorem}

\paragraph{Sum of $\ell_{0}$-Sampling Sketches via Dynamic Tree.} Suppose our goal is to maintain a data structure for a (not necessarily spanning) forest $F$ of a graph $G$. We next review a result that allows us to a build a data structure on $F$ such that for any connected component $S$ of $F$, we can compute $\Pi(b_S)$ in $\Otil(1)$ time, which is faster than the previous approach that yielded an $\Otil(|S|)$ time algorithm by explicitly summing $\sum_{u \in V} \Pi(B_u)$. This construction is implicit in \cite{kapron2013dynamic} and can be alternatively obtained by combining Theorem 4.16 of \cite{NanongkaiS17} with \Cref{thm:sketch} and \Cref{prop:update sketch}. 

\begin{theorem}
\label{thm:sketch tree} Let $\Pi$ be a linear transformation for the $\ell_0$-sampling sketch from~\Cref{thm:sketch}. There is a data structure $\mathcal{D}(\Pi,G,F)$ that maintains an $n$-vertex graph $G=(V,E)$ and a (not necessarily spanning) forest $F$ on the same vertex set $V$ that supports the following operations
\begin{itemize}
\itemsep0em
\item insert or delete an edge in $G$ in $\Otil(1)$ time,
\item insert or delete an edge in $F$ (as long as $F$ remains a forest) in $\Otil(1)$ time, and
\item given a pointer to a connected component $S$ of $F$,  return $\Pi(b_S)=\sum_{u\in V}\Pi(b_{u})$, where $b_{u}$ is the $u$-th row of the incidence matrix $B$ of $G$, in $\Otil(1)$ time.
\end{itemize}
\end{theorem}



\subsection{The Algorithm}

In this section, we show an algorithm for maintaining the edge connectivity of an $n$-vertex dynamic graph undergoing edge insertions and deletions in $\Otil(n)$ worst-case update time with high probability, i.e., we prove \Cref{thm:main n}. 

Let $G$ be a graph undergoing edge insertions and deletions. We make the following two simplifying assumptions
\begin{enumerate}
\itemsep0em
    \item[(1)] the edge connectivity is attained at a \emph{non-singleton} minimum cut $C^{*}$ of $G$, and
    \item[(2)] the minimum degree $\delta$ of $G$ is between some fixed range $[\delta_0/2, \delta_0]$, and we shall construct a data structure depending on $\delta_0$.
\end{enumerate}

We lift (1) by observing that if this assumption doesn't hold, then we have that the edge connectivity $\lambda = \delta$, and the minimum degree $\delta$ of $G$ can be maintained in a straightforward way. Assumption (2) can be lifted by constructing $O(\log n)$ data structures
for the ranges $[2^{i},2^{i+1}]$ for $i \leq \log n+1$, and query the data structure for range $[2^{i},2^{i+1}]$ whenever $\delta\in[2^{i},2^{i+1}]$.



Before we proceed further, recall that $R$ is a random 2-out subgraph and $G' = G/R$ is the random 2-out contraction of $G$. By \Cref{thm:random-2-out}, the minimum cut $C^{*}$ is preserved in $G'$ with some positive constant probability. We can boost this to a high probability bound by repeating the whole algorithm $O(\log n)$ times. As $G'$ contains $O(n/\delta_{0})$ vertices
with high probability, throughout we will also assume that this is the case.

Despite the size guarantee on the vertex set of $G'$, maintaining $G'$ and running a static edge connectivity algorithm on $G'$ is not enough as the number of edges in $G'$ could potentially be large. This leads us to the main component of our data structure that instead maintains a $\delta_0$-connectivity certificate $H'$ of $G'$ while supporting updates as well as querying $H'$ in $\Otil(n)$ update time. 

\begin{theorem} \label{thm:maintainCertificate} Let $G$ be a dynamic graph and let $G'$ be the random $2$-out contraction of $G$. Let $\delta_0$ be an integer parameter such that  $n \ge \delta_0 > \lambda$. There is a data structure that supports edges updates to $G$ (and thus to $G'$)  and gives  query access to a $\delta_0$-connectivity certificate $H'$ of $G'$ containing $\Otil(n)$ edges with high probability. The updates can be implemented in $\Otil(\delta_0)$ worst-case time,
the  queries in $\Otil(n)$ worst-case time. 
\end{theorem}


We claim that \Cref{thm:maintainCertificate} proves \Cref{thm:main n}. To see this, recall that $H'$ preserves all cuts of size at most $\delta_0$ in $G'$ by the definition of $\delta_0$-connectivity certificate. By our assumption $\lambda < \delta_0$, this implies that the minimum cut $C^{*}$ is preserved in $H'$ and suggests that we can simply run a static edge connectivity algorithm on top of $H'$ to find $C^{*}$ in $\Otil(|E(H')|)=\Otil(n)$ time. Therefore, the rest of this section is devoted to proving~\Cref{thm:maintainCertificate}.

The first component of the data-structure is to maintain (i) a random 2-out subgraph $R$ of $G$ and (ii) a spanning forest $F$ of $R$. By the definition of $R$, (i) can be implemented in $O(1)$ worst-case update time per edge update. For (ii), we use the dynamic spanning forest data structure of Kapron, King and Mountjoy~\cite{kapron2013dynamic} that guarantees a $\Otil(1)$ worst-case update time. Since for each edge update to $G$, there are only $O(1)$ edge updates to $R$, this can be implemented in $\Otil(1)$ worst-case time per edge update as well.
We use  the following two-step approach to prove our result
\begin{enumerate}
\itemsep0em
\item[(1)] reducing the dynamic $\delta_0$-connectivity certificate problem to the dynamic $O(\log n)$-connectivity certificate problem via the template presented in \Cref{alg:para certificate}, and
\item[(2)] solving the dynamic $O(\log n)$-connectivity certificate problem using the linear sketching tools developed in the previous section.
\end{enumerate}


\paragraph{Reducing $\delta_{0}$-Connectivity Certificates to 
$O(\log n)$-Connectivity Certificates.}

We follow the template of \Cref{alg:para certificate} from~\Cref{thm:para certificate}:
Set $k=\delta_{0}$ and let $c=k/(4\tau\ln n)$ and $k'=\left\lceil 6\tau\log n\right\rceil$
be defined as in \Cref{alg:para certificate}. Then, color each edge of $G$ by randomly choosing a color from $\{1,\dots,c\}$. Let $G_{i}$
be the subgraph of $G$ containing edges of color $i$. We observe that
all graphs $G_{i}$ can be maintained explicitly together with $G$ with in $\Otil(1)$ time per edge update. Similarly, let $G'_{i}$ be the subgraph of the random $2$-out contraction $G'$ containing edges of color $i$, i.e., $G'_{i} = G' \cap G_i$. 

Our goal is to not explicitly maintain $G'_{i}$, but instead build a dynamic data structure with $\Otil(1)$
worst-case update time per edge update in $G$ that gives query access to a $k'$-connectivity certificate $H'_{i}$ of $G'_{i}$ in $\Otil(n/\delta_{0})$
time with high probability. 

We claim that this suffices to prove~\Cref{thm:maintainCertificate}. To see this, note that for each $i=1,\ldots,c$ we can query $H'_i$ in $\tilde{O}(n/\delta_0)$ time. Then we simply union all these certificates to compute $H'=\cup_{i=1}^{c}H'_{i}$ in $c \cdot \Otil(n/\delta_{0})= \delta_0/(4\tau \ln n) \cdot  \Otil(n/\delta_0) = \Otil(n)$ time, which bounds the query time. To bound the update time, note that the worst-case cost for maintaining these $c$ data structures is $c \cdot \Otil(1) = \Otil(\delta_0)$. By \Cref{thm:para certificate} it follows that (i) $H'$ is indeed a $\delta_{0}$-connectivity certificate $H'$ of $G'$ and (ii) the size of $H'$ is at most $\delta \cdot O(n/\delta) = O(n)$, which completes the proof of~\Cref{thm:maintainCertificate}.


\paragraph{$O(\log n)$-Connectivity Certificates of $G'_{i}$ via Linear Sketching.}

Let us fix a color $i = 1,\ldots,c$. Recall that our goal is to obtain a $O(\log n)$-connectivity
certificate $H'_{i}$ of $G'_{i}$ in time $\Otil(n/\delta_0)$ per query and $\Otil(1)$ per edge update.

Recall that $G_{i}$, which is the graph containing all color-$i$ edges of $G$, is explicitly maintained. Let $F$ be a spanning
forest of the random 2-out subgraph $R$ which we also maintain explicitely as discussed above. For all $j, \ell = 1, \ldots, O(\log n)$,
we independently generate the linear transformation for a  $\ell_{0}$-sampling
sketch $\Pi^{(j,\ell)}$ using \Cref{thm:sketch}. For each index pair
$(j,\ell)$, we build a data structure ${\cal D}^{(j,\ell)}:={\cal D}(\Pi^{(j,\ell)},G_{i},F)$\footnote{Here we deliberately omit the subscript $i$ from ${\cal D}^{(j,\ell)}$ since the color $i$ was fixed in the beginning of the paragraph and this omission simplifies the presentation.} and maintain whenever $G$ or $F$ changes
using \Cref{thm:sketch tree}.  Since for each edge update to $G$, there are only $O(1)$ edge updates to $G_{i}$ and $F$, respectively, we can maintain all ${\cal D}^{(j,\ell)}$'s in $\Otil(1)$ worst-case
update time per edge update to $G$. This completes the description of handling edge updates and its running time analysis. 

It remains to
show how to query a $O(\log n)$-connectivity certificate $H'_{i}$
of $G'_{i}$ using ${\cal D}^{(j,\ell)}$ in $\Otil(n/\delta_{0})$
time.
The main idea of our construction is the following simple but crucial observation. Each vertex $u$
in $G'_{i}$ corresponds to a tree $S_{u}$ of $F$ as $S_{u}$ is contracted into $u$. The latter holds by the definition of $G'_i$. Therefore, if we let $B$ and $B'$
denote the incidence matrices of $G_{i}$ and $G'_{i}$, respectively,
then we have $b_{S_{u}}=b'_{u}$. Now, using the data structures ${\cal D}^{(j,\ell)}$
we can retrieve $\Pi^{(j,\ell)}(b_{S})$ for all components
$S$ of $F$, which in turn gives us $\Pi^{(j,\ell)}(b'_{u})$
for all $u\in V(G'_{i})$. Note that the total time to query ${\cal D}^{(j,\ell)}$ for all
$j,\ell = 1, \ldots, O(\log n)$ and all $u \in V(G'_i)$ is $\Otil(|V(G'_{i})|)=\Otil(n/\delta_{0})$.

Now, the sketches $\Pi^{(j,\ell)}(b'_{u})$
for all $u\in V(G'_{i})$ and all $j,\ell = 1,\ldots,O(\log n)$ allow us to compute $O(\log n)$-connectivity certificate
$H'_{i}$ of $G'_{i}$ as shown below. Note that the algorithms modifies the sketches $\Pi^{(j,\ell)}(b'_{u})$ but we may revert the sketches back to their initial state after computing and returning $H'_i$.

We finally explain the procedure for querying $H'_{i}$. To this end, we follow the template of \Cref{alg:seq certificate}
from \Cref{thm:seq certificate}. Set $G_{1}^{tmp}\gets G'_{i}$ to be
the temporary graph that we will work with. Consider the $j$-th round of the algorithm where $j = 1, \ldots, O(\log n)$.
We compute a spanning forest $F_{j}$ of $G_{j}^{tmp}$ using the
sketches $\Pi^{(j,1)},\dots,\Pi^{(j,\log n)}$
on vertices of $G'_{i}$ in $\Otil(|V(G'_{i})|)=\Otil(n/\delta_{0})$
time using \Cref{thm:forest from sketch}. Next, we update $G_{j+1}^{tmp}\gets G_{j}^{tmp}\setminus F_{j}$
and also update the sketches $\Pi^{(j',\ell)}$ for all
$j'>j$ so that they maintain information of graph $G_{j+1}^{tmp}$
and not of $G_{j}^{tmp}$. This takes $\Otil(n/\delta_{0})$ time because
$F_{j}$ contains $\Otil(n/\delta_{0})$ edges. This ends the $j$-th round. After all rounds have been completed, we return
$H'_{i}=\cup_{j}F_{j}$ which is a $O(\log n)$-connectivity certificate
by \Cref{thm:seq certificate}. Since there are $O(\log n)$ iterations, the total query time is $\Otil(n/\delta_{0})$, what we wanted to show. Note that algorithm internally stores the edge connectivity value only and not the edges on a cut that attains the edge connectivity. Therefore, the adversary cannot reveal anything useful from querying this information, and thus it follows that our algorithm works against an adaptive adversary. 

\section{Deterministic Algorithm with $O(m^{1-\eps})$ Update Time}
In this section we prove \Cref{thm:main det}. Our algorithm requires several tools from different works and we review them below.
\label{sec:det}

\subsection{Algorithmic Tools}

\paragraph{Expander Decomposition and Pruning.}
The key to our approach is the notion of \emph{expander} graphs. The \emph{conductance} of an unweighted, undirected graph is defined by 
\[
\Phi(G) \defeq \min_{\emptyset \subset S \subset V } \partial(S)/\min\{\vol_G(S), \vol_G(V \setminus S)\}.
\]
We say that a graph $G$ is $\phi$-\emph{expander} if $\Phi(G) \geq \phi$.
Next, we introduce the notion of expander decomposition.
\begin{definition}[Expander Decomposition] \label{def: expanderDecomp}
Let $G=(V,E)$ be an undirected, unweighted graph and let $\phi \in (0,1)$ be a parameter. A \emph{$\phi$-expander decomposition} of $G$ is a \emph{vertex-disjoint} partitioning $\mathcal{U}= \{U_1,\ldots,U_k\}$ of $V$ such that
\begin{itemize}
\itemsep0em
\item  $\sum_{1 \leq i \leq k} \abs{\partial(U_i)} = O(\phi m)$, and
\item  for each $1 \leq i \leq k$,  $\Phi(G[U_i]) \geq \phi$.
\end{itemize}
\end{definition}

We now review an efficient algorithm for finding expander decompositions.
\begin{theorem}[\cite{SaranurakW19}] \label{thm:expanderDecomp} Let $G=(V,E)$ be an undirected, uniweighted graph and let $\phi \geq 0$ be a parameter. There is an algorithm \textsc{Expander}$(G,\phi)$ that in $\tilde{O}(m/\phi)$ time finds a $\phi$-expander decomposition.
\end{theorem}

The next result allows us to turn static expander decompositions into their dynamic variants, as shown in Section~\ref{sec: decrementalExpanderDecom}

\begin{theorem}[Expander Pruning, \cite{SaranurakW19}] \label{thm:pruning} Let $G=(V,E)$ be an undirected, unweighted $\phi$-expander. Given an online sequence of $k \leq \phi \vol(G)/20$ edge deletions in $G$, there is an algorithm that maintains a \emph{pruned set} $P \subseteq V$ satisfying the following properties; let $G_i$ and $P_i$ be the graph $G$ and the set $P$ after the $i$-th deletion. For $1 \leq i \leq k$,
\begin{enumerate}
\itemsep0em
\item $P_0 = \emptyset$ and $P_i \subseteq P_{i+1}$,
\item $\vol_G(P_i) \leq 8i/\phi$ and $|\partial(P_i)| \leq 4i$, and \label{pruning:volume}
\item $G_i[V \setminus P_i]$ is a $\phi/6$-expander. \label{pruning:Expansion}
\end{enumerate}
The total time for updating $P_0, \ldots, P_k$ is $O(k \log m / \phi^2)$.
\end{theorem}

\paragraph{Non-singleton Minimum Cut Sparsifiers.}

Our algorithm relies on sparsifiers that preserve non-singleton cuts of simple graphs.

\begin{definition} Let $G=(V,E)$ be an undirected, unweighted graph. A multi-graph $H=(V',E')$ is a \emph{non-singleton minimum cut sparsifiers} \emph{(}abbrv. \emph{\sparsifier}\emph{)} of $G$ if $H$ preserves all non-singleton minimum cuts of $G$, i.e., for all cuts $C \subset V$ with $\min\{|C|, |V \setminus C|\} \geq 2$ and $\partial_G(C) = \lambda(G)$, 
\[
		\abs{\partial_H(C)} = \abs{\partial_G(C)}.
\]
We say that $H$ is of size $k$ if $|E'| \leq k$. 
\end{definition}

Kawarabayashi and Thorup~\cite{KawarabayashiT19} showed that undirected, simple graphs admit \sparsifier{s} of size $\Otil(m/\delta)$. They also designed a deterministic $\tilde{O}(m)$ time algorithm for computing such sparsifiers. We will construct \sparsifier{s} that are based on expander decompositions, following the work of Saranurak~\cite{Saranurak21}, as we can turn them into a dynamic data structure. To do so, we need to formally define the procedures of \emph{trimming} and \emph{shaving} vertex subsets of a graph. 
\begin{definition}[Trim and Shave] \label{def: TrimShave}
Let $U$ by any vertex subset in $G=(V,E)$. Define $\trim(U) \subseteq U$ to be the set obtained by the following procedure: while there exists a vertex $u \in U$ with $|E(u,U)| < \frac{2}{5} d(u)$, remove $u$ from $U$. Let $\shave(U) = \{u \in U \mid |E(u,U)| > \frac{1}{2} d(u) + 1\}$. 
\end{definition}

Observe that the trimming procedure recursively removes a vertex with few connections inside the \emph{current} set $U$ while the shaving procedure removes all vertices with few connections inside the \emph{initial} set $U$. Saranurak~\cite{Saranurak21} showed that we can construct \sparsifier{s} by  applying triming and shaving  to each cluster in the expander decomposition. We formally summarize his construction in the lemma below. 
\begin{lemma}[\cite{Saranurak21}] \label{lem: NMCsparsifier} Let $G=(V,E)$ be an undirected, simple graph with $m$ edges, and let $\phi = c/\delta$, where $\delta$ is the minimum degree in $G$ and $c \geq 40$ some positive constant. Let \[ \mathcal{U} = \textsc{Expander}(G, \phi),~\mathcal{U}' = \{\trim(U) \mid U \in \mathcal{U} \} \text{ and } \mathcal{U}'' = \{\shave(U') \mid U' \in \mathcal{U'}\} .\] Let $H=(V',E')$ be the graph obtained from $G$ by contracting every set $U'' \in \mathcal{U}''$. Then $H$ is an \sparsifier~of size $\tilde{O}(\phi m)$ for $G$. The running time for computing $H$ is $\tilde{O}(m)$.
\end{lemma}

\subsection{Decremental Expander Decomposition} \label{sec: decrementalExpanderDecom}

In this section we show that the expander pruning procedure from Theorem~\ref{thm:pruning} allows us to design a dynamic algorithm for maintaining an expander decomposition under edge deletions. While the theorem below is already implicit in other works leveraging the power of expander decompositions~\cite{GoranciRST21,BernsteinBGNSS022} (albeit with slightly different guarantees and variations depending on the specific application), here we give a simple, self-contained version that suffices to solve the edge connectivity problem.

\begin{theorem} \label{thm: decrementalExpanderDecomp} Given an unweighted, undirected graph $G=(V,E)$ with $m$ edges and a parameter $\phi \in (0,1)$, there is a decremental algorithm that supports an online sequence of up to $\phi^2 m$ edge deletions and maintains a $\phi/6$-expander decomposition in $\tilde{O}(m/\phi)$ total update time.
\end{theorem}

We initialize our data structure by (i) constructing a $\phi$-expander decomposition $\mathcal{U}$ of the initial graph $G$ using Theorem~\ref{thm:expanderDecomp}, where $\phi \in (0,1)$ is a parameter, and (ii) starting a pruning data-structure $\DSprune(U)$ for each expander $U \in \mathcal{U}$ using Theorem~\ref{thm:pruning}. We also maintain a counter $\counterDeletion(U)$ that denotes the number of edge deletions inside the cluster $U$. Initially, $\counterDeletion(U) = 0$ for each $U \in \mathcal{U}$. If the total number of edge deletions exceeds $\phi^2 m$, our data-structure terminates.

We next show how to handle edge deletions. Consider the deletion of edge $e$ from $G$. If $e$ is an inter-cluster edge in $\mathcal{U}$, then we simply remove it from $G$ since its removal does not affect the expansion of any of the clusters in $\mathcal{U}$. Otherwise, $e$ is an intra-cluster edge, and let $U_e \in \mathcal{U}$ be the unique cluster that contains $e$. First, we increase the counter  $\counterDeletion(U_e)$ by $1$. Next we compare the number of deletions in the cluster relative to the number of deletions the pruning procedure can handle.

Concretely, if $\counterDeletion(U_e) \leq \phi \vol_G(U_e)/ 20$, we pass the deletion of $e$ to the pruning data structure $\mathcal{D}(U_e)$. Let $P_i$, resp., $P_{i-1}$  be the pruned set that $\mathcal{D}(U_e)$ maintains after, resp. before the deletion $e$. We define $S = \{\{u\} \mid u \in P_i \setminus P_{i-1}\}$ to be the set of singleton clusters, and then replace $U_e$ in $\mathcal{U}$ with $\{U_e \setminus \{ P_i \setminus P_{i-1} \}, S\}$. The last step can be thought of as including every vertex in $P_i \setminus P_{i-1}$ as a singleton cluster in $\mathcal{U}$ and removing these vertices from the current expander $U_e$.

However, if $\counterDeletion(U_e) > \phi \vol_G(U_e)/ 20$, then we declare every vertex in the current cluster $U_e$ to be a singleton cluster. Specifically, we remove $U_e$ from $\mathcal{U}$ and for each $u \in U_e$, add $\{u\}$ to $\mathcal{U}$. Note that the latter implies that all vertices that belonged to the \emph{original} cluster $U_e$ are included as singletons in the current expander decomposition. This completes the description of the procedure for deleting an edge. 

We next show that the above algorithm correctly maintains an expander decomposition under edge deletions while paying a small constant factor in the expansion guarantee of each cluster and in the number of inter cluster edge.

\begin{lemma}[Correctness] \label{lem: DecExpDecompCorrectness}
The decremental algorithm maintains a $\phi/6$-expander decomposition.
\end{lemma}
\begin{proof}
Let $\mathcal{U}$ be expander decomposition that the algorithm maintains for the current graph $G$. 

Our first goal is to show that  for each $U \in \mathcal{U}$, $\Phi(G[U]) \geq \phi/6$. Observe that by construction each cluster $U$ in $\mathcal{U}$ can either be (i) a singleton, (ii) a pruned cluster (i.e., a cluster that is formed by removing vertices from the original cluster) or (iii) an original cluster from the initial expander decomposition. If a cluster is a singleton, then the expansion bound trivially holds. If we have a type (ii) cluster, then by expander pruning (Theorem~\ref{thm:pruning}, Property~\ref{pruning:Expansion}), it follows that $\Phi(G[U]) \geq \phi/6$, where $G$ is the current graph. Finally, for a type (iii) cluster, the initial expander decomposition~(Theorem~\ref{thm:expanderDecomp}) gives that $\Phi(G[U]) \geq \phi \geq \phi/6$. Combining the above cases, leads to the expansion bound we were after.

We now bound the number of inter cluster edges in $\mathcal{U}$. Recall that initially, the expander decomposition has at most $O(\phi m)$ inter-cluster edges~(Theorem~\ref{thm:expanderDecomp}). During the handling of edge deletions, the algorithm introduces new inter-cluster edges when vertices from the pruned set are included as singletons. Thus, our ultimate goal is to bound the volume of the pruned set with the number of edge deletions in a cluster. To this end, let $U$ be a cluster in $\mathcal{U}$. We distinguish two cases. If $\counterDeletion(U) \leq \phi \vol_G(U)/20$ then, then by expander pruning~(Theorem~\ref{thm:pruning}, Property~\ref{pruning:volume}) we have that the maintained pruned set $P_U$ satisfies $\vol(P_U) \leq 8 \cdot  \counterDeletion(U)/\phi$. However, if $\counterDeletion(U) > \phi \vol_G(U)/20$, then by construction, the pruned set $P_U$ is the entire original cluster $U$. By rearranging the inequality in the last condition, we get $\vol(P_U) = \vol(U) \leq 20 \cdot \counterDeletion(U)/\phi$. Combining the above bounds, we get that at any time during our decremental algorithm, the volume of the maintained pruned set of a cluster satisfies $\vol(P_U) = \vol(U) \leq 20 \cdot \counterDeletion(U)/\phi$.

Summing this over all clusters in the expander decomposition $\mathcal{U}$, we have the number of the new inter-cluster edges is bounded by
\[
	\sum_{U \in \mathcal{U}} \vol(P_U) \leq \frac{20}{\phi} \sum_{U \in \mathcal{U}} \counterDeletion(U) \leq \frac{20}{\phi} \phi^2 m = O(\phi m),
\]
where the penultimate inequality follows from the fact the number of edge deletions to $G$ is bounded by $\phi^2 m$ by the assumption of the lemma. Thus, the number of inter-cluster edges increases by a constant multiplicative factor, which concludes the proof of the lemma. 
\end{proof}

We next bound the running time of our decremental algorithm.

\begin{lemma}[Running Time]  \label{lem: DecExpDecompRunTime}
The decremental expander decomposition runs in $\tilde{O}(m/\phi)$ total update time.
\end{lemma}
\begin{proof}
The running time of the algorithm is dominated by (1) the time required to compute the initial expander decomposition and (2) the total time to perform expander pruning on each cluster of this decomposition. By Theorem~\ref{thm:expanderDecomp}, (1) is bounded by $\tilde{O}(m/\phi)$. By Theorem~\ref{thm:pruning},  the pruning process on a cluster $U$ can be implemented in $\tilde{O}(\phi \vol_G(U) \cdot \log m/ \phi^2) = \tilde{O}(\vol_G(U)/\phi)$. Summing over all the clusters in the expander decomposition $\mathcal{U}$ and recalling that they form a partition of $V(G)$, we get that the running time of (2) is bounded by
\[
	\sum_{U \in \mathcal{U}} \tilde{O}(\vol_G(U)/\phi) = \tilde{O}(m/\phi).
\] 
Bringing together (1) and (2) proves the claim of the lemma.
\end{proof}

\subsection{Fully Dynamic \sparsifier} \label{sec: fullyDynamicSparsifier}
In this section present a fully dynamic algorithm for maintaining a \sparsifier~of undirected, simple graphs.

\subsubsection{Decremental \sparsifier}
We start by showing that the decremental expander decomposition~(Theorem~\ref{thm: decrementalExpanderDecomp}) almost immediately yields a decremental algorithm for maintaining a \sparsifier.
More specifically, we show the following theorem in this subsection.
\begin{theorem} \label{thm: decrementalSparsifier} Given an unweighted, undirected graph $G=(V,E)$ with $m$ edges and a parameter $\phi \in (0,1)$ satisfying $\phi \geq c/\delta$ for some positive constant $c \geq 240$, there is a decremental algorithm that supports an online sequence of up to $\phi^2 m$ edge deletions and maintains a \sparsifier~$H$ of size $\tilde{O}(\phi m)$ in $\tilde{O}(m/\phi)$ total update time.
\end{theorem}

Let $\phi \in (0,1)$ be parameter with $\phi \geq c/\delta$ for some positive constant $c \geq 240$. Our data-structure internally maintains an expander decomposition under edge deletions~\textsc{DecExpander}$(G,\phi)$~(Theorem~\ref{thm: decrementalExpanderDecomp}). Let $\mathcal{U}_{(0)}$ be the expander decomposition of the initial graph $G^{(0)}$ from \textsc{DecExpander}$(G,\phi)$. Let $\mathcal{U}_{(0)}'= \{\trim(U) \mid U \in \mathcal{U}_{(0)}\}$ and $\mathcal{U}_{(0)}'' = \{\shave(U') \mid U' \in \mathcal{U}'_{(0)}\}$. We define $H^{(0)}=(V',E')$ to be the graph obtained from $G^{(0)}$ by contracting every set $U'' \in \mathcal{U}_{(0)}''$. As we will shortly see, $H^{(0)}$ will correspond to a \sparsifier~of $G^{(0)}$. This suggests that in order to maintain such a sparsifier under edge deletions we need to efficiently maintain the sets $U'= \trim(U)$ and $U'' = \shave(U')$ for every cluster $U$ in the current expander decomposition $\mathcal{U}$. We achieve this by keeping track of the following counters:
\begin{itemize}
\itemsep0em
\item the degree $\deg(u)$ of each vertex $u$ in $V$, and 
\item the degree $\deg_{U'}(u) := |E(u,U')|$ for all $u \in U'$ and all $U' \in \mathcal{U'}$, i.e., the degree of vertex $u$ restricted to the cluster $U'$.
\end{itemize}
Note that both degree values can be computed for each vertex in the initial graph $G_0$ by performing a graph traversal.

Now, consider the deletion of an edge $e=(u,v)$ from $G$. We first decrement the value of both counters $\deg(u)$ and $\deg(v)$ by one to account for the deletion of $e$. Then we pass this deletion to the data-structure $\textsc{DecExpander}(G,\phi)$. This in turn reports a subset of vertices $P_U$ that are pruned out of a cluster $U$ due to the deletion of $e$. At this point observe that the decremental expander decomposition algorithm already has updated $U$ with respect to $P_U$.  Thus it remains to update the sets $U'$ and $U''$ respectively. 

For each $u \in P_U$ we do the following. First, note that when $u \not \in U'$, we don't need to do anything since $U'' \subseteq U'$ asserts that $u$ cannot belong to the contracted set $U''$. If $u \in U'$, then we remove $u$ from $U'$ and potentially $U''$ by invoking the subprocedure $\textsc{Remove}(u)$ defined as follows:
\begin{itemize}
\itemsep0em
\item[] \textsc{Remove}$(u)$
\item[] $\quad \bullet~$ Set $U'_{\mathrm{old}} \gets U'$, and set $U' \gets U' \setminus \{u\}$.
\item[] $\quad \bullet~$ If $u \in U''$ then 
\begin{itemize}
\itemsep0em
\item[] $\quad \triangleright~$ Set $U'' \gets U'' \setminus \{u\}$, and $\textsc{Uncontract}(u)$
\end{itemize}
\item[] $\quad \bullet~$ For every neighbour $v \in E(u, U'_{\mathrm{old}})$:
\begin{itemize}
\itemsep0em
\item[] $\quad \triangleright~$ Set $\deg_{U'}(v) \gets \deg_{U'}(v) - 1$.
\item[] $\quad \triangleright~$ If $\deg_{U'}(v) < 2/5 \deg(v)$ then 
\begin{itemize}
\itemsep0em

\item[] $\quad -~$ $\textsc{Remove}(v)$ 
\end{itemize}
\item[] $\quad \triangleright~$ If $\deg_{U'}(v) <  \deg(v) /2 + 1$ and  $v \in U''$ then
\begin{itemize}
\itemsep0em
\item[] $\quad -~$  Set $U'' \gets U'' \setminus \{v\}$, and $\textsc{Uncontract}(v)$.
\end{itemize}
\end{itemize}
\end{itemize}

Procedure \textsc{Uncontract}$(u)$ simply reverts the operation of contracting $u$ into some cluster $U''$. It can also be interpreted as adding the vertex $u$ together with its incident edges in $G$ to the current sparsifier $H$. This completes the description of the algorithm.

The next lemma shows that the algorithm maintains a sparsifier that preserves non-singleton minimum cuts exactly.

\begin{lemma}[Correctness] \label{ref: correctnessDecSparsifier}
The decremental algorithm correctly maintains a \sparsifier~$H$ of size $O(\phi m)$.
\end{lemma}
\begin{proof}
We begin by showing that $H$ is a \sparsifier~of some current graph $G$. To this end, let $G^{(i)}$ 
be the graph after the $i$-th deletion and let $H^{(i)}$ be the sparsifier after the data-structure has processed the $i$-th deletion. To prove that $H^{(i)}$ is a \sparsifier~of $G^{(i)}$ it suffices to show that (i) $\mathcal{U}^{(i)}$ is $\phi/6$-expander decomposition of $G^{(i)}$ (ii) ${\mathcal{U}_{(i)}}'= \{\trim(U) \mid U \in \mathcal{U}_{(i)}\}$, and (iii) $\mathcal{U}_{(i)}'' = \{\shave(U') \mid U' \in \mathcal{U}_{(i)}'\}$. To see why this is true, note that $\phi/6 \geq 240/(6\delta) \geq 40/\delta$ by assumption of Theorem~\ref{thm: decrementalSparsifier}, and apply Lemma~\ref{lem: NMCsparsifier} with the parameter $\phi/6$.

If $i=0$, recall that by construction $\mathcal{U}_{(0)}$ is a $\phi$-expander decomposition of the initial graph $G^{(0)}$, $\mathcal{U}_{(0)}'= \{\trim(U) \mid U \in \mathcal{U}_{(0)}\}$ and $\mathcal{U}_{(0)}'' = \{\shave(U') \mid U' \in \mathcal{U}'_{(0)}\}$. Since the parameter $\phi$ satisfies $\phi \geq 240/\delta \geq 40/\delta$, we get the graph $H^{(0)}$ obtain by contracting every set $U'' \in \mathcal{U}''_{(0)}$ is a \sparsifier~of $G^{(0)}$. 

If $i \geq 1$, inductively assume that $\mathcal{U}_{(i-1)}, \mathcal{U}'_{(i-1)}, \mathcal{U}''_{(i-1)}$ have been correctly maintained until the $(i-1)$-st edge deletion. By Theorem~\ref{thm: decrementalExpanderDecomp}, we already know that $\mathcal{U}_{(i)}$ is a $\phi/6$-expander decomposition of the graph $G^{(i)}$. Thus it remains to argue about the correctness of $\mathcal{U}'_{(i)}$ and $\mathcal{U}''_{(i)}$.

Let $P_{U}$ be the set of vertices pruned out of a cluster $U$ that the data-structure $\textsc{DecExpander}(G,\phi)$ returns upon the $i$-th edge deletion. To prove that the update of $U'=\trim(U)$ and $U''=\shave(U')$ with respect to $P_U$ is correct, by Definition~\ref{def: TrimShave}, consider the following invariants
\begin{itemize}
\itemsep0em
\item[(1)] for all $u \in U' = \trim(U), ~ \deg_{U'}(u) \geq 2/5 \deg(u)$.
\item[(2)] $U'' = \shave(U') =  \{u \in U' \mid \deg_{U'}(u) \geq \deg(u)/2 + 1 \}$.
\end{itemize}

For every vertex $u \in P_U$ with $u \in U'$,  note that our subprocedure $\textsc{Remove}(u)$ removed from $U'=\trim(U)$ all vertices $v$ for which $\deg_{U'}(v) < 2/5 \deg(v)$, and thus the invariant (1) holds for the vertices that are left in $U'$. 
Moreover, as already pointed out in~\cite{chuzhoy2019new}, $\trim(U)$ is unique so the order in which vertices are removed does not matter. Similarly, by construction we have that subprocedure $\textsc{Remove}(u)$ detects all vertices in $U''=\shave(U')$ that do not satisfy invariant (2). It follows that $U'$ and $U''$ are maintained correctly, which in turn implies that $\mathcal{U}'_{(i)}$ and $\mathcal{U}''_{(i)}$ are also correct. 

The guarantee on the size of the sparsifier follows directly from Lemma~\ref{lem: NMCsparsifier}.
\end{proof}

We next study the running time of our algorithm.

\begin{lemma}[Running Time] The decremental algorithm for maintaining a \sparsifier~ runs in $\tilde{O}(m/\phi)$ total update time.
\end{lemma}
\begin{proof}
The running time of the algorithm is dominated by (1) the time to maintain a decremental expander decomposition $\mathcal{U}$, (2) the total time to maintain $\mathcal{U}'$ and $\mathcal{U}''$ and (3) the cost of performing vertex uncontractions. By Theorem~\ref{thm: decrementalExpanderDecomp}, (1) is bounded by $\tilde{O}(m/\phi)$. To bound (2), we can implement subprocedure $\textsc{Remove}(u)$ for a vertex $u$ in $O(\deg(u))$ time, excluding the recursive calls to its neighbours. Since the updates from $G$ and $\mathcal{U}$ are decremental (as they consist of either edge deletions or vertex deletions), once a vertex leaves a set $U'$ or $U''$, it can never join back. Hence, it follows that (2) is bounded by $\sum_{u \in V}O(\deg(u)) = O(m)$. Similarly, a vertex $u$ can be uncontracted at most once, and this operation can also be implemented in $O(\deg(u))$ time, giving a total runtime of $\sum_{u \in V}O(\deg(u)) = O(m)$ for (3). Bringing (1), (2) and (3) together proves the claim of the lemma.
\end{proof}

\subsubsection{Extension to Fully Dynamic \sparsifier}

We follow a widespread approach in data structures for turning a decremental algorithm into a fully dynamic algorithm and apply it to our problem for maintaining a \sparsifier.

At a high level, our approach uses a decremental algorithm for maintaining a \sparsifier~of the graph and handles edge insertions by keeping them ``on the side". It crucially relies on the fact that adding an edge to the sparsifier yields a sparsifier for the new graph augmented by that edge. To make sure that the size of the sparsifier remains small after these edge augmentations, we restart our decremental algorithm from scratch after the number of updates exceeds a predefined threshold. This leads to the following result.

\begin{theorem} \label{thm:fullyDynamicSparsifier}
Given an unweighted, undirected graph $G=(V,E)$ with $m$ edges and a parameter $\phi \in (0,1)$ satisfying $\phi \geq c/\delta$ for some positive constant $c \geq 240$, there is a fully dynamic algorithm that maintains a \sparsifier~$H$ of size $\tilde{O}(\phi m)$ in $\tilde{O}(1/\phi^{3})$ amortized time per edge insertion or deletion.
\end{theorem}

Our data structure  subdivides the sequence of edge updates into phases of length $\phi^2 m$, where $\phi \in (0,1)$, satisfying $\phi\geq c/\delta$ for some positive constant $c \geq 240$. Our algorithm maintains
\begin{itemize}
\item the set of edges $I$ that represents the edges inserted since the beginning of a phase that have not been subsequently deleted.
\end{itemize}

At the beginning of each phase, we initialize (i) the decremental algorithm \textsc{DecSparsifier$(G,\phi)$} (Theorem~\ref{thm: decrementalSparsifier}) to maintain a \sparsifier~$H$ of the current graph $G$, and (ii) set $I \gets \emptyset$.  

Let $e$ be an edge update to $G$. If $e$ is an edge insertion to $G$, we add it to the set $I$. If $e$ is deleted from $G$, we consider two cases: If $e \in I$, we simply delete $e$ from $I$. If $e \not \in I$, we pass the deletion of $e$ to $\textsc{DecSparsifier}(G,\phi)$ to update the sparsifier $H$. We maintain $H' \gets H \cup I$ as the sparsifier of the current graph. This completes the description of the algorithm. 

We next show that our fully dynamic algorithm maintains a correct \sparsifier~at any time.

\begin{lemma}[Correctness] \label{ref: correctnessFDSparsifier}
The fully dynamic algorithm correctly maintains a \sparsifier~$H'$ of size $O(\phi m)$.
\end{lemma}
\begin{proof}
Let $G=(G_0 \setminus D) \cup I$ be the current graph, where $G_0$ is the graph at the beginning of a phase, $D$ is the set of edges deleted from $G_0$, and $I$ is the of edges inserted since the beginning of a phase that have not been subsequently deleted. Let $H' = H \cup I$ by the sparsifier our data-structure maintains.

By Theorem~\ref{thm:fullyDynamicSparsifier}, we know that $H$ is a \sparsifier~of $G_0 \setminus D$. We claim that $H \cup I$ is a \sparsifier~of $(G_0 \setminus D) \cup I$. To see this, consider the case when $I = \{e\}$, where $e=(u,v)$. Once proving this simpler case, our general claim follows follows by induction.  As $H$ is a contraction of $G_0 \setminus D$ (Lemma~\ref{lem: NMCsparsifier}), there is a vertex mapping $f: V(G_0 \setminus D) \rightarrow V(H)$ assigning nodes that are contracted together to a single node in $V(H)$. We distinguish two cases. If $f(u) \neq f(v)$, then $e$ increases a non-singleton minimum cut in $G_0 \setminus D \cup \{e\}$ by at most one. Since the edge is present both in $G_0 \setminus D \cup \{e\}$ and $H \cup \{e\}$, it follows $H \cup \{e\}$ preserves all non-singleton minimum cuts of $G_0 \setminus D \cup \{e\}$. If $f(u) = f(v)$, i.e., both endpoints of $e$ are contracted into a single vertex in $H$, then we claim that $e$ cannot participate in any non-singleton minimum cut in $G_0 \setminus D \cup \{e\}$. Suppose for contradiction that the latter holds. Then there exists a non-singleton minimum cut $(C,V \setminus C)$ in $G \setminus D \cup \{e\}$ such that
\[
	|E_{G \setminus D \cup \{e\}}(C, V \setminus C)| = |E_{G \setminus D} (C, V \setminus C)| + 1,
\]
where the above equality uses the fact that $e$ is a cut edge. Since a non-singleton minimum cut $G \setminus D$ can increase by at most $1$ when adding a single edge, it follows that $(C, V \setminus C)$ is a non-singleton minimum cut in $G$. Since $H$ is \sparsifier~of $G$, we have that $(C,V \setminus C)$ must also be non-singleton minimum cut in $H$. Let $s \in V(H)$ be the supervertex in $H$ containing $u$ and $v$. It follows $\min\{|s \cap C|, |s \setminus C| \}  \neq 0$, which contradicts the fact that $H$ is \sparsifier~of $G$.

To bound the size of $H'$, observe that (1) $H$ is of size $O(\phi m)$ at any time~(Theorem~\ref{thm: decrementalSparsifier}), and (2) $|I| \leq \phi^{2} m$. Therefore, $H'$ is of size $O(\phi^2 m + \phi m) = O(\phi m).$
\end{proof}

\begin{lemma}[Running Time] \label{ref: runTimeFDSparsifier}
The fully dynamic algorithm for maintaining a \sparsifier~runs in $\tilde{O}(1/\phi^3)$ amortized time per edge insertion or deletion.
\end{lemma}
\begin{proof}
The total update to maintain a decremental sparsifier is $\tilde{O}(m/\phi)$~(Theorem~\ref{thm: decrementalSparsifier}) under the condition that the number of deletions is smaller then $\phi^2 m$. Our data-structure makes sure that the number of updates within a phase never exceeds $\phi^2 m$. Charging the total update time to these updates, we get an amortized update time of $\tilde{O}(m/\phi \cdot 1/(\phi^2 m)) = O(1/\phi^3)$.
\end{proof}

\subsection{The Algorithm}

In this section we show an algorithm for Theorem~\ref{thm:main det}. Our main idea is to run in ``parallel" a variant of the fully algorithm for maintaining a \sparsifier~(Theorem~\ref{thm:fullyDynamicSparsifier}) and the exact fully dynamic edge connectivity algorithm due to Thorup~\cite{thorup2007fully} which is efficient whenever the edge connectivity is polylogarithmic or a small polynomial in the number of vertices. We also maintain a carefully chosen threshold edge connectivity value which tells us when to switch between the two algorithms.
We start by observing that the fully dynamic algorithm for maintaining a \sparsifier~$H$ of a graph $G$~(Theorem~\ref{thm:fullyDynamicSparsifier}) gives the following simple algorithm for edge connectivity: (i) maintain the minimum degree $\delta$ of the current graph $G$, (ii) after each edge update compute $\lambda(H)$ on the graph $H$~(Theorem~\ref{thm: staticEdgeConnectivity}), and (iii) set the edge connectivity $\lambda(G)$ of the current graph to be $\min\{\delta, \lambda(H)\}$. The following corollary is an immediate consequence of Theorem~\ref{thm:fullyDynamicSparsifier}.

\begin{corollary} \label{cor: fullyDynamMC}
Given an unweighted, undirected graph $G=(V,E)$ with $m$ edges and a parameter $\phi \in (0,1)$, there is a fully dynamic algorithm for maintaining an edge connectivity estimate $\mu(G)$ in $\tilde{O}(1/\phi^3 +\phi m)$ amortized time per edge insertion or deletion. If $\phi \geq c/\delta$, for some positive constant $c \geq 240$, then the edge connectivty estimate is correct, i.e., $\mu(G) = \lambda(G)$.
\end{corollary}

Next we review the result of Thorup~\cite{thorup2007fully}  concerning efficient maintenance of small edge connectivity.

\begin{theorem}[Theorem 26,~\cite{thorup2007fully}] \label{thm: fullyDynamThorup}
Given an unweighted, undirected graph $G=(V,E)$ with $n$ edges, and a parameter $\eta \in (0,n)$, there is a fully dynamic algorithm for maintaining an edge connectivity estimate $\mu(G)$ in $\tilde{O}(\eta^{29/2} \sqrt{n})$ worst-case time per edge insertion or deletion. If $\lambda(G) \leq \eta$, then the edge connectivity estimate is correct, i.e., $\mu(G) = \lambda(G)$.
\end{theorem}

We now have all the tools to present our sub-linear fully dynamic edge connectivity algorithm, which proceeds as follows. Let $\tau=n^{a}$ be a threshold value on the edge connectivity to be determined shortly, where $a \in (0,1)$ is a parameter. We run
\begin{itemize}
\itemsep0em
\item[] (1) the fully dynamic algorithm $\mathcal{A}_1$ from Theorem~\ref{thm: fullyDynamThorup} with parameter $\eta = (\tau +1)$, and
\item[] (2) the fully dynamic algorithm $\mathcal{A}_2$ from Corollary~\ref{cor: fullyDynamMC} with parameter $\phi = 240/\tau$.
\end{itemize}
We extend both algorithms $\mathcal{A}_1$ and $\mathcal{A}_2$ to perform a test on how edge connectivity $\lambda(G)$ of the current graph compares to the threshold value $\tau$ \emph{after} the algorithm that is currently being used to answer queries has processed an edge update. These extensions allow us to switch between these two algorithms, so the queries we answer regarding $\lambda(G)$ are correct.

First, observe that both algorithms internally explicitly maintain $\lambda(G)$. We proceed as follows
\begin{itemize}
\itemsep0em
    \item Suppose $\mathcal{A}_1$ is currently being used to answer queries. If $\lambda(G) \leq \tau$ after an update operation operation, then we do not switch. Otherwise (i.e., $\lambda(G) = (\tau + 1)$), we switch to $\mathcal{A}_2$ for the next operation.
    \item Suppose $\mathcal{A}_2$ is currently being used to answer queries. If $\lambda(G) \geq (\tau+1)$ after an update operation, then we do not switch. Otherwise (i.e., $\lambda(G) = \tau$), we switch to $\mathcal{A}_1$ for the next operation.
\end{itemize}

We next prove the correctness. It suffices to verify that our parameter requirements from \Cref{thm: fullyDynamThorup} and \Cref{cor: fullyDynamMC} are satisfied whenever we use one of the algorithms to answer queries. Let $G$ be the current graph. Note that if $\lambda(G)$ was at most $\tau$ before an update operation, we used algorithm $\mathcal{A}_1$, which works correctly, even then $\lambda(G)$ reaches $\tau + 1$ after that operation. If $\lambda(G)$ was at least $\tau+1$ before the operation, we use $\mathcal{A}_2$ which works correctly even if $\lambda(G)$ drops to $\tau$ after the operation as $\tau \le\lambda(G)$. In either case we have that
$\phi = 240/\tau \ge 240/\lambda(G) \ge 240/\delta$. This completes the correctness proof.

The running time is bounded as follows. By Theorem~\ref{thm: fullyDynamThorup} and $\lambda(G) \leq \tau +1$, (1) supports edge updates in $\tilde{O}(\tau^{29/2} \sqrt{n})$ worst-case time. By Corollary~\ref{cor: fullyDynamMC}, (2) guarantees a $\tilde{O}(1/\phi^3 + \phi m) = \tilde{O}(\tau^3 + m/\tau)$ amortized time per update. Bringing these running times together, we get that the amortized time per edge update is
\[
	\tilde{O}(\tau^{29/2} \sqrt{n} + \tau^3 + m/\tau) = \tilde{O}(\tau^{29/2} \sqrt{n} + m/\tau).
\]

Balancing the two terms in the above expression, we get $\tau = m^{\nfrac{2}{31}}/n^{\nfrac{1}{31}}$, which in turn implies that the amortized update time of our algorithm is $O(m^{\nfrac{29}{31}}n^{1/31})$. This completes the proof of Theorem~\ref{thm:main det}.

\section{Concluding remarks and open problems}
We showed two sub-linear algorithms for exactly maintaining edge connectivity in fully dynamic graphs. The main idea behind both algorithms was to maintain sparsifiers that preserve non-singleton cuts dynamically, and this was achieved by leveraging the power of random 2-out contractions and expander decompositions in the context of edge connectivity. 

Our work leaves several natural open problems.

\begin{itemize}
\itemsep0em
    \item[(1)] \emph{Can our update time for dynamically maintaining exact edge connectivity be improved?} We remark that a closer examination of our result based on expander decompositions reveals that an improvement to Thorup’s result~\cite{thorup2007fully} for bounded edge connectivity (specifically, improving the polynomial dependency on the edge connectivity) would immediately lead to an improved running time. It would be very interesting to investigate whether this can be achieved.

    \item[(2)] \emph{Is there a fully dynamic algorithm for $(1+\eps)$-approximating edge connectivity in $\textup{poly}(\log n)\eps^{-O(1)}$ update time?} The best-known algorithm due to Thorup achieves $\tilde{O}(\sqrt{n})$ update time, and even going beyond this $\tilde{O}(n^{1/2})$ barrier remains an important open problem in dynamic graph algorithms.
\end{itemize}

\section*{Acknowledgments}
We thank Tijn de Vos and Aleksander Bjørn Grodt Christiansen for pointing out an error when citing Thorup~\cite{thorup2007fully} in a previous version of the paper.

\small
\setlength\parskip{0cm}
\setlength\itemsep{0cm}
\bibliographystyle{plain}
\bibliography{references}


\end{document}